\documentclass[
	12pt, 
	a4paper, 
	]{article}
\usepackage{amssymb} 
\usepackage{amsmath} 
\usepackage{amsfonts} 
\usepackage{eurosym} 
\usepackage{geometry} 
\usepackage{graphicx} 
\usepackage{enumitem}
\usepackage{bbm}
\usepackage{subfigure} 
\usepackage{caption}
\usepackage{float}
\usepackage{color} 
\usepackage{setspace} 
\usepackage{footmisc} 
\usepackage{natbib} 
\usepackage{pdflscape} 
\usepackage{array} 
\usepackage{hyperref} 
\usepackage{ntheorem} 
\usepackage{tikz}
\usepackage{lscape}
\usepackage{comment}
\usepackage{mathtools}
\usepackage{dsfont}
\usepackage{multimedia}
\usepackage{url}
\usepackage{lineno}
\usetikzlibrary{shapes,snakes}
\usetikzlibrary{arrows.meta}
\usepackage{varwidth}
\usetikzlibrary{positioning}
\usetikzlibrary{shapes.multipart}

\newcommand{\R}{\color{red} }
\newcommand{\B}{\color{blue}}
\theoremstyle{break} 
\newtheorem{definition}{Definition}

\newtheorem{theorem}{Theorem}

\newtheorem{lemma}{Lemma}
\newtheorem{proposition}{Proposition}

\newenvironment{proof}[1][Proof]{\noindent\textbf{#1. }}{\ \rule{0.5em}{0.5em}}

\newcolumntype{L}[1]{>{\raggedright\let\newline\\arraybackslash\hspace{0pt}}m{#1}}
\newcolumntype{C}[1]{>{\centering\let\newline\\arraybackslash\hspace{0pt}}m{#1}}
\newcolumntype{R}[1]{>{\raggedleft\let\newline\\arraybackslash\hspace{0pt}}m{#1}}

\begin{document}

\begin{titlepage}
\title{Staking Pools on Blockchains\thanks{We would like to thank Jihao Zhang, Philémon Bordereau, Christian Matt, Christopher Portmann, Daniel Tschudi, Christian Ewerhart, Stefanos Leonardos, Jiasun Li, Luyao Zhang,  participants of the Conference on Mechanism and Institution Design 2022 and the 33rd Stony Brook International Conference on Game Theory for valuable feedback.}}

 \author{
 	Hans Gersbach\\
 	\normalsize CER-ETH and CEPR\\
 	\normalsize Z\"{u}richbergstrasse 18\\
 	\normalsize 8092 Zurich, Switzerland\\ 
    \normalsize \href{mailto:hgersbach@ethz.ch}{hgersbach@ethz.ch}
    \and
 	Akaki Mamageishvili\\
 	\normalsize CER-ETH\\
 	\normalsize Z\"{u}richbergstrasse 18\\
 	\normalsize 8092 Zurich, Switzerland\\ 
	\normalsize \href{mailto:amamageishvili@ethz.ch}{amamageishvili@ethz.ch}
 	\and
 	Manvir Schneider\\
 	\normalsize CER-ETH\\
 	\normalsize Z\"{u}richbergstrasse 18\\
 	\normalsize 8092 Zurich, Switzerland\\ 
 	\normalsize \href{mailto:manvirschneider@ethz.ch}{manvirschneider@ethz.ch}
 	}

\date{September 2022}

\maketitle

\begin{abstract}
On several proof-of-stake blockchains, agents engaged in validating transactions can open a pool to which others can delegate their stake in order to earn higher returns. We develop a model of staking pool formation in the presence of malicious agents and establish existence and uniqueness of equilibria. We then identify potential and risk of staking pools. First, allowing for staking pools lowers blockchain security. Yet, honest stake holders obtain higher returns. Second, by choosing welfare optimal distribution rewards, staking pools prevent that malicious agents receive large rewards. Third, when pool owners can freely distribute the returns from validation to delegators, staking pools disrupt blockchain operations, since malicious agents attract most delegators by offering generous returns.
	\vspace{0.2in}
	 
\end{abstract}

\thispagestyle{empty}
\end{titlepage}
\setcounter{page}{2}


\section{Introduction}

From both the financial side and the security side, there are reasons why a proof-of-stake (PoS) blockchain may want to allow the formation of staking pools. With staking pools, agents interested in validating transactions are allowed to open a pool,  such that others can delegate their stake for some time to it.   
For delegating agents who are not interested in validating transactions, this can provide an additional income on their token holdings. In turn, their stakes are blocked and cannot be used for other purposes during the time of commitment. 
By agents we mean all participants of the decentralized system who own some stake, which usually takes the form of a native token.  
Agents interested in running a staking pool could earn a higher income from their transaction validation activities.\footnote{An agent that decides to run a pool and validate transactions is usually referred as a node of a network.} 

Ideally, such a staking system makes it more attractive to hold tokens, provides incentives for a sufficient number of agents to run staking pools and act as transaction validators, and  increases the share of honest agents, weighted by the stakes they hold, involved in transaction validation.  As every staking pool acts as a validator, we occasionally use the word ``validator'' for a staking pool.

However, malicious agents also run staking pools and may thus enlarge the share of the stake they control in transaction validation.\footnote{Our assumption that malicious agents always run staking pools is justified, since if they do not, they do not affect blockchain functioning and thus cannot be considered malicious. If they delegate their stakes, they act like honest delegators and thus again are not malicious in any way.} This may undermine the security of the blockchain and lead to a collapse of the protocol, as the malicious agents take over. 

We explore how such a system can be modeled and designed, so that it operates beneficially for the decentralized consensus mechanism---i.e. by lowering the share of malicious agents who disrupt the validation of transactions---and for the ecosystem as a whole. The model invokes a measure of honest agents who are interested in the returns from holding a stake (of tokens) in a PoS blockchain and thus are also interested in the proper functioning of the blockchain. An agent is honest  if s/he is prepared to run the software for validation, as required by the system. 
Otherwise, honest agents choose actions to maximize their expected returns.
The PoS protocol that we consider in this paper perfectly mimics how the proof-of-work (PoW) protocol works. The PoW protocol chooses the next block writer according to who finds the nonce that satisfies certain conditions, that is, the probability that the next writer is chosen proportionally to his/her {\it hash rate}. Similarly, in the PoS protocol, the next writer of the block is chosen proportionally to his/her stake size. In the PoW protocols, the costs are typically assumed to be different across agents, as they depend on electricity, hardware, and maintenance costs. In the PoS protocol, we have a similar situation, except for electricity costs. Agents have different costs in participating in transaction validation, as availability of appropriate computer software and hardware, speed and bandwidth of the internet, the knowledge of how to run a secure validation node, and opportunity costs to engage in validation activities(that is, costs of not being able  to use tokens for other purposes) differ across them. The difference between PoW and PoS protocol pools is the following. In the former, the costs are incurred by all pool members, and therefore the rewards are distributed proportionally. In the latter, the costs are incurred only by the pool runner, while the rest---the delegators---incur no cost. In this paper, we examine the simplest reward distribution, where the pool runner keeps some fraction of the rewards and distributes the rest to the delegators proportionally. 

There is a measure of malicious agents who are only interested in disrupting the blockchain, and therefore, their costs are ignored by the designer. 

A {\it Blockchain Designer} aims at maximizing the chance that the blockchain is working (maximizing blockchain security), which will be captured by maximizing the number of honest validators. 
We will also consider an alternative objective where the Blockchain Designer trades off the probability that the blockchain is running correctly with the costs for all honest agents of validating transactions. This is a standard economic welfare criterion.
  Two further aspects can be important for a Blockchain Designer: Reducing the rewards for malicious agents, as this decreases their future influence and distributing the rewards to validators as equally as possible. 

The Blockchain Designer has two basic options when designing the market for staking pools. First, s/he can fix the return distribution between the pool owner and the pool delegators. We call this ``return fixing".
Second, s/he can allow competition of pool runners regarding how the returns from transaction validation are shared between the pool owner and the pool delegators. This is called ``return competition". 

We model the ensuing interaction as a three-stage game. In the first stage, agents decide whether (i) to open a staking pool, (ii) to delegate their stake to some pool or (iii) to abstain from validation activities. Setting up pools for validating transactions is costly, and these costs may differ between agents.  
In the second stage, either the Blockchain Designer determines the shares uniformly for all running pools (return fixing) or pool owners determine how returns should be shared between pool owners and delegators (return competition). In the third stage, transactions are validated and, depending on the share of stakes controlled by malicious agents, validation either works properly or the blockchain is disrupted. 

Our main insights start from the observation that honest agents with high costs to set-up a node as a validator may want to delegate their stake to other pool owners, while honest agents with low costs may want to open their own pool. 
Malicious agents always open a pool, as this increases their chances to disrupt the blockchain.

We establish existence and uniqueness of equilibria of the stake pool formation game with fixed return distribution between pool owners and delegators and show that they are of the threshold type. We show next that there exists a unique sharing rule of the returns from validation between delegators and pool owners that maximizes the probability that the blockchain operates correctly and we do the same for maximizing welfare of honest agents when costs of running staking pools are taken into account in addition to blockchain security. Then, we provide numerical illustrations of the equilibria for uniform distribution of costs. 

Subsequently, we identify the potential and risk of staking pools.
First, staking pools can never increase current blockchain security over a system in which no such pools are allowed. The reason is as follows. Without staking pools,  a share of honest agents  participates in validating transactions, as the additional reward is higher than the costs. With staking pools and if the returns are shared with the validators, the return to pool owners declines, as the average rewards from validating transactions is given. Hence, less honest agents are willing to open staking pools, so that malicious agents will control a larger share of stakes in validating transactions~\footnote{The return splitting cannot be strictly enforced by the protocol (designer) itself, but should rather be a recommendation. Therefore, if some pool owner publicly offers more rewards to delegators than  is recommended, it should be a signal that this pool owner is not honest.}. 

Yet, by optimally choosing the distribution of the validation returns to delegators, the allocation of rewards to honest stakes involved in validation increases, which may be beneficial for subsequent blockchain operations. We show how return splitting between pool owners and delegators has to be determined in order to minimize the rewards to malicious agents. 

Second, by choosing welfare optimal distribution rewards, staking pools may decrease blockchain security, but it prevents allocating large rewards to only a fraction of agents.

 Third, when pool owners can freely distribute the returns from validation to delegators, staking pools decrease blockchain security, since malicious agents attract delegators by distributing most of the returns to them.  

Finally, we show how our results can be extended to situations in which not only running a pool but also the act of delegation is costly. Moreover, as we use a continuum model for tractability reasons, we show how our analysis can be recast---albeit with a more complex formal apparatus---in a discrete setting.  

The paper is organized as follows. In the next section, we discuss the related literature. In particular, literature motivating the formation of staking pools in PoS blockchains is reviewed. In Section~\ref{sec: model}, we introduce the model and preliminaries. In Section~\ref{sec: eq_analysis}, we analyze the equilibria of the fixed return game. In Section~\ref{sec: reward_design}, we discuss designs for staking rewards that either maximize security or maximize welfare. Section~\ref{sec: numerical} addresses a uniform cost distribution and provides numerical examples. In Section~\ref{sec: competition}, we analyze a return competition game, where pool owners individually decide on the reward sharing scheme. In Section~\ref{sec: extensions}, we study extensions of our basic model. Section \ref{sec: conclusion} concludes.

\section{Related Literature}\label{sec: literature}

 \textit{Staking Pools:} 
 Many blockchains have already implemented staking or will implement it in the near future. Such examples include, but are not limited to Cardano (\cite{ouroboros}), Solana (\cite{Yakovenko2017}), Polkadot (\cite{Wood2016}), Tezos (\cite{Goodman2014}) and Concordium. 
 All these allow staking pools in which agents who do not run their own staking pool will be able to delegate their stake to an existing pool and benefit from rewards. By delegation, agents are indirectly involved in block proposal and validation, via their stake.
 
\cite{Brunjes2020} study staking pools among honest agents from an interesting mechanism design perspective. Their reward scheme ensures that a desired number of staking pools is achieved while each pool has approximately the same amount of stake and low-cost agents are running the pools. The reward scheme ensures that reporting the true costs is the dominant strategy. 
Our paper is complementary as we focus on the design of staking pools in the presence of malicious agents who want to disrupt the blockchain and thus have quite different objectives than honest agents. We examine on how such staking pools affect blockchain security, how security risks can be alleviated and how distribution of rewards to malicious agents can be limited. Our mechanism is also simpler to implement, as there is no communication between the designer and pool runners, and therefore, no need for contracting, unlike in~\cite{Brunjes2020}. It is also intuitive to interpret for the agents, than the generic mechanisms studied in there.

\textit{Blockchain and Consensus Protocols:} There is a rapidly growing literature on blockchains and consensus protocols. The most common reference is~\cite{Nakamoto2008}, a Bitcoin whitepaper, in which the consensus protocol PoW is introduced. Many papers focusing on consensus protocols, that have been published recently, include~\cite{John2020} and~\cite{Benhaim2021}. PoS protocols were introduced to avoid the significant use of energy associated with  PoW protocols. In this paper, we adopt a widely used PoS protocol, in which the next validator is chosen proportionally to the available stake.  
A  major paper on the game-theoretic analysis of blockchain and its folk theorem  is~\cite{Biais2019}.~\cite{Benhaim2021} and~\cite{Amoussou2020} consider games in the presence of malicious agents. This aspect is also part of our model. In particular, we argue that the probability that the blockchain is well-functioning is increasing with the share of honest validators and achieves complete security (probability one) when the share of malicious agents falls below a certain threshold. {More specific work on blockchain mining rewards from a game-theoretic angle are studied in~\cite{Chen2019} and~\cite{Kiayias2016}.}

\textit{Model Assumptions:}
\cite{Herrera2014} study a voting game with two parties where voters vote for one party or abstain. In their model, voters have individual costs that are drawn according to some distribution function. Similarly to our paper,~\cite{Herrera2014} characterize equilibria of their game which take the form of so called ``cut-off thresholds''. They obtain a pair of thresholds, one for each party. Then, a citizen whose cost is below the corresponding threshold will turn out and vote for his/her party. If a citizen's cost is above the threshold, s/he will abstain. Our model works in a similar manner. We characterize threshold equilibria such that if the cost for an individual is below that threshold, s/he will run a staking pool, and  delegate or abstain otherwise.

Similar to~\cite{Herrera2014},~\cite{Castanheria2003} studies a voting game where citizens' costs are drawn from a uniform distribution. In our paper, we also consider a uniform distribution for costs. One key difference, however, is that the costs in~\cite{Herrera2014} and~\cite{Castanheria2003} are costs of voting\footnote{See the rational and costly voting literature (\cite{Palfrey1983}, \cite{Ledyard1984}, \cite{Borgers2004} and \cite{Gersbach2021a}).}, whereas in our paper, the costs are associated with running a pool and the  costs connected to it.

Apart from direct mining rewards for the next block miner, there are also rewards coming from auctions, where users bid on their transactions to be included to the next block. The rewards from auctions may sometimes even be higher than mining rewards, see~\cite{MEV_rewards} for the extended treatment of maximal extractable value (MEV) and its role on the security of blockchain protocols. There are even calls for aligning rewards from MEV with rewards from mining~\footnote{https://notes.ethereum.org/cA3EzpNvRBStk1JFLzW8qg\#Committee-driven-smoothing}, which would result into scaling up the reward parameter of our paper.

In our paper, we use continuum approach to model the measure of agents, and thus follow the approach in~\cite{Gersbach2009} and~\cite{Halaburda2021}, for instance. More concretely, agents are modeled as infinitely small. The continuum approximates large communities and it proves to be a tractable approach for the staking pool formation game. The continuum model is thus a limiting case where the number of agents becomes large and delegation to staking pools is done uniformly at random. In our model, besides a pool ID (or address), no further information is provided to delegators. From addresses or pool IDs, no information about pool owners can be inferred. Hence, every pool has equal chances to be chosen by agents.\footnote{In practice, delegators may have more information about pool owners, but they remain anonymous. On the Cardano blockchain, for example, agents can find all staking pools on \url{pool.pm}, which visualizes all staking pools with the pool IDs, the current stake of the pools, the number of delegators and which blocks were produced by which staking pool.} In the basic version of the model, delegation entails no cost. 

 \textit{Crypto-Democracy and Vote Delegation:} There is an extensive research on democracy and blockchains, such as voting on blockchain (see~\cite{Leonardos2020},~\cite{Osgood2016} and~\cite{Allen2017}, for instance). Very recent work on vote delegation in the presence of malicious agents is~\cite{Gersbach2021b}. The delegation of votes can be seen as a delegation of stakes in the blockchain environment. Further literature in the field of vote delegation is known under ``liquid democracy", where besides voting and abstaining, agents have the additional option to delegate their vote to other agents (see~\cite{Bloembergen2019}, for example).

 \textit{Automated Market Maker (AMM):} Staking is also used in other contexts such as yield aggregation\footnote{Under yield aggregation, investors can passively earn rewards by sending tokens to reward generating smart contracts (see~\cite{Cousaert2021}).} and automated market maker. The analogy to staking pools is the following. There are liquidity providers who add funds to a liquidity pool that is managed by an AMM. The liquidity providers correspond to agents who delegate their stake in our model. Hence, the AMMs collect funds in a liquidity pool that can be traded later.\footnote{Note that the liquidity pools hold at least two different assets, which is a major difference to staking pools.} In return, liquidity providers receive tokens proportional to their staked amount, which correspond to the rewards in our model. Besides these analogies, AMMs have other tasks as well, such as determining prices for traders. A well-known AMM is Uniswap, for example (see~\cite{Adams2020}).
 
 \textit{PoW and Mining Pools:} PoW has been studied extensively since the rise of Bitcoin \cite{Nakamoto2008}. While in PoS protocols, the next validator is chosen randomly based on their amount of stakes, PoW protocols use computational power. That is, the chance that a miner will mine the next block is proportional to his/her computational power\footnote{As a measure of computational power, one typically uses the hash rate.} relative to the total computational power of all miners. Similar to PoS protocols, miners in PoW protocols can form a pool to gather more computational power and hence increase the chance to mine the next block. Such mining pools have been analyzed game-theoretically in~\cite{Lewenberg2015},~\cite{Leonardos2019},~\cite{Cong2020},~\cite{Cheung2021},~\cite{Arnosti2022} and~\cite{Chatzigiannis2022}, for example. 
 ~\cite{Cong2020} study mining pools and focus on risk sharing as centralizing force. In their model, miners are modelled as a continuum who can invest a hash rate.~\cite{Arnosti2022} study a model where miners have heterogeneous costs and show that asymmetries in costs lead to concentration of mining power. Different levels of electricity costs, for example, are a source of heterogeneity of costs. In PoS heterogeneity of costs of validators to run a pool can arise because of different skills and computer facilities to set up a validator node and offering it as a continuously running validator and pool for other stake holders and different operation costs of validation. ~\cite{Chatzigiannis2022} models miners as rational agents who distribute their power across pools and across different cryptocurrencies. Reward functions for mining pools are studied in e.g.~\cite{Schrijvers2017}.~\cite{Fisch2017} study optimal mining pooling strategies in PoW blockchains.

\section{Model}\label{sec: model}

\subsection{The General Set-up}

There is a continuum of measure $H$ of \textit{honest} agents\footnote{These agents are called ``rational" by other authors, e.g. in \cite{Amoussou2020} and~\cite{Halaburda2021}.} and there is a continuum of measure $M$ of \textit{malicious} agents. The continua can be represented by intervals on the real line, with length $H$ and $M$, respectively. Working with a continuum of agents models a blockchain with a large number of participants and approximates the corresponding discrete model. Indeed, we show in Section~\ref{Discrete_Case}, how the analysis---albeit in a much more complicated way---can be performed in a discrete framework and how we can recover results in that framework. While the  continuum model is much more tractable and yields simpler expressions, it requires more subtle interpretations, though.

We assume $H > M$, so that, honest agents are in the majority. Each agent (malicious or honest) has one unit of the stake.\footnote{As the total amount of stakes is infinite, all variables which are integrated over the set of agents are averages in the continuum model.}  

Each honest agent is identified by his/her cost level for validation of transactions, respectively, for running a staking pool on the blockchain.\footnote{Costs include, for example, the costs for registering, running the software and forwarding messages on transactions.} Costs are denoted by $c$ and are heterogeneous across honest agents, as they depend on the availability of appropriate computer capital and human capital. Let the random variable $X$ correspond to the costs for honest agents. Specifically, the costs for honest agents are distributed according to the atomless density function $f(c)$ defined on $[0, T)$. Note that the support interval can be $\mathbb{R}_+$, that is, $T$ can be equal to $\infty$.  The corresponding cumulative distribution function is denoted by $F(c)$. 
Malicious agents are of the Byzantine type and do not care about the costs and returns of running a pool. Hence, we set their costs to zero.

There is also a reward $R\in \mathbb{R}_+$, paid for creating the next block\footnote{In some contexts it is called mining reward.}. The agents' types are private information. As assumed above all honest and malicious agents have the same amount of stakes, equal to one unit. There is also a \textit{Blockchain Designer}. The blockchain is assumed to be functioning better if more of validators are honest.

\subsection{Objectives}

The Blockchain Designer and the two types of agents have the following general objectives:

\begin{itemize}
    \item Maximize the chance that the blockchain is working (maximizing blockchain security), which will be captured by a  maximization of the number of honest validators (\textbf{Blockchain Designer}). 
    \item Maximize expected reward minus cost (\textbf{Honest Agents}).
    \item Maximize the measure of stakes delegated to them (\textbf{Malicious Agents}).
\end{itemize}

 We will formally specify  the  manifestation of these objectives later. As to the Blockchain Designer, we will consider an alternative objective where s/he trades-off the probability that the blockchain is running correctly with the costs for all honest agents to validate transactions. 
 While we consider maximizing blockchain security as the most important objective, arguably one could also consider a standard economic welfare criterion as a guiding principle for the Blockchain Designer. 
 
 Two further aspects can be important for a Blockchain Designer. First, s/he may aim at minimizing the rewards received by malicious agents, as this decreases their future influence. Second, the Blockchain Designer may want to distribute rewards to validators as equally as possible, which is an original motivation of staking. We will discuss to which extent these aspects materialize when we present our results.  

\subsection{Staking Pool Formation Game}

We consider the following game, which consists of three stages: 

\begin{itemize}
    
    \item[ ] Stage 1: Agents decide either to form a staking pool or not (both honest and malicious). Agents who decide to become a pool owner obtain an identification number, denoted by $i$.
    
    \item[ ] Stage 2: Agents who did not register for a staking pool decide whether to delegate their stake to some staking pool or to remain idle. 
    
    \item[ ] Stage 3: The blockchain runs, validation takes place (or not), and rewards are distributed. 
\end{itemize}

If an agent $i$ forms a staking pool, we denote by $s_i$ the amount of stakes s/he is receiving. We also denote by $P$ the measure of honest agents who form a staking pool. $D$ and $I$ denote the measure of honest agents who delegate their stakes or stay idle, respectively. We have $H = P + D + I.$

Our main assumption for this game is that delegators distribute themselves evenly across all possible pools. The rationale is that  the type of a pool owner is private information, and for delegators, pool owners  are all alike. Hence, invoking measure consistency, this assumption implies
\begin{equation*}
    s_i = \frac{D}{P+M}, \forall i.
\end{equation*} 

We note that all pools obtain the same amount of delegated stakes and thus we write $s$ for $s_i$ in the following. The total size of the pool---stakes of the pool owner and delegated stakes---is then $s+1$. We note that $(s+1)(P+M) + I = D + P + M + I = M + H$.

The payoffs are determined as follows:
The next validator\footnote{In blockchains, transaction validation is done by creating a new block.} is chosen among the available pools proportionally to the pool size. This particular rule is already implemented in major PoS protocols (e.g., see~\cite{ouroboros}), as it perfectly replicates the PoW protocol, where probability for the next leader (in our case next validator who writes the block) is proportional to  computing powers. Such a system has the advantage that splitting and pooling of the stakeholders does not increase the chances to be chosen as a next validator as the expected return is unaffected by such strategies.

The reward for the next block is given by an amount $R$. Since all pools have the same size, the return distribution is a uniform distribution with density $\frac{1}{P+M}$. Hence, the individual reward a pool expects to receive is $r= \frac{R}{P+M}$. Since we have a continuum model, we note that both the individual return for an individual and the cost of running a pool have zero weight in the average return $R$ and the average amount of costs, respectively.  Yet for an agent, only the individual returns and costs matter.    

The blockchain designer sets a parameter, denoted by $\lambda$, $0 < \lambda \leq 1$, which determines how rewards have to split between pool owners and delegators.
To sum up, the individual expected rewards are as follows: 

\begin{itemize}
    \item A pool receives $r= \frac{R}{P+M}$.
    
    \item The pool owner obtains $\lambda \cdot  r$.
    
    \item An individual delegator obtains $\frac{(1-\lambda) \cdot  r}{s}$. The total amount given to the delegators in a pool is $(1-\lambda) \cdot  r$.
    
    \item An idle agent obtains 0. 

\end{itemize}

\noindent We measure the probability that the blockchain operates correctly by a function $$P_c(\cdot):[0,T]\rightarrow [0,1],$$ that depends on the share of honest agents running staking pools, that is (weakly) increasing as a function of this share, and which may reach probability one if a sufficient share of honest agents is participating in validating transactions. 
Later in this paper, we will study two different versions of this probability function and reward schemes that depend on it. 
In the basic version of the staking pool formation game, we assume that the returns are paid, no matter whether the blockchain operates correctly or not. The motivation for this assumption is as follows: Whether or not the blockchain operates correctly may not be immediately detected or agents maybe able to sell their rewards immediately after writing the next block. Hence, agents involved in validating transactions aim  at maximizing the immediate returns from these activities in such cases. The Blockchain Designer is, of course, interested in how well the blockchain is functioning.  
In Section~\ref{sec: endogenous_rewards}, however, we make rewards dependent on the operation of blockchains, that is, rewards are only distributed if the blockchain operates correctly.

The main design parameter $\lambda \in [0,1]$ is a non-negative real number set by the Blockchain Designer. 
A game with a payoff structure as above, together with $(H,M,R,F,\lambda)$, is called a ``staking pool formation game'' and is denoted by $\mathcal{G}$.

\section{Equilibrium Analysis}\label{sec: eq_analysis}
In this section, we analyze the equilibria of a staking pool formation game.

\subsection{Equilibrium Concept}
In a staking game with non-zero reward $R$, an honest agent will never stay idle. The reason is that since delegation is free of cost and delegators earn some reward, the expected payoff is positive for a delegator, whereas the payoff for staying idle is zero. Only in the case $\lambda = 1$, agents would be indifferent between delegation and staying idle. Hence, delegation weakly dominates staying idle and so $H = P + D$.

In Section~\ref{sec: costlydelegation}, we will consider the case where delegation comes at a small fixed cost and so, staying idle will not be dominated in general. For tractability we assume a tie-breaking rule. First, if an honest agent is indifferent between delegating and staying idle, s/he will choose to delegate. This implies that when there are no delegation costs, no agent remains idle.

We now proceed by focusing on equilibria of the threshold type. In particular, we solve for the threshold equilibrium by looking at the agent with a specific cost level $c^*$ at which this agent is indifferent between running a pool and delegating, i.e., the expected utility from delegating is equal to the expected utility from running a pool. 


An important remark is in order. If a threshold equilibrium exists, it is unique. If a threshold equilibrium does not exist, we end up in a corner solution---either no honest agent will run a pool or all honest agents will run a pool, and thus there will be no delegation. 

We introduce the following definition: 

\begin{definition}[Threshold Equilibrium]
A cost level $c^*>0$ is called a ``threshold equilibrium'' if an agent with cost $c^*$ is indifferent between running a staking pool and delegating. Furthermore, all agents incurring a cost that is lower than $c^*$ will run a staking pool, and all agents with a cost greater than $c^*$ will delegate. 
\end{definition}

In the threshold equilibrium we have $P=F(c^*)H$ and $D=(1-F(c^*))H$.

\subsection{Equilibrium Characterization}

 In the following we are looking for the equilibria, in which at least some fraction of honest agents decide to run own pools. 
We  characterize the equilibria of the staking game: 

\begin{theorem}\label{lower_bound_lambda}
There exists a unique threshold equilibrium  to the game $\mathcal{G}$ if and only if \begin{equation}
\label{lambda_lowerbound}
    \lambda >\frac{M}{H+M}
    . 
\end{equation}

\end{theorem}

\begin{proof}
First, we have to set up the indifference condition for an honest agent. That is, we have to equate the expected utility from being a delegator with the expected utility from running a pool. More precisely, the expected utility from being a delegator is $$\frac{(1-\lambda)r}{s},$$ which is the share $(1-\lambda)$ of the reward $r$, divided by the number of delegators, for the particular pool.
Similarly, the expected utility from running a pool is $$\lambda r - c,$$ which is the share $\lambda$ of the reward $r$, minus the cost $c$ of the particular pool owner. In the equilibrium point $c^*$, an honest agent is indifferent between delegating and running a pool, that is, $c^*$ solves the indifference equation:
\begin{equation}\label{indifference_original}
    \frac{(1-\lambda)r}{s} = \lambda r -c^*.
\end{equation}
From the equilibrium definition, we know that $P=F(c^*)H$ and $D=(1-F(c^*))H$. Plugging in these values in~\eqref{indifference_original} and simplifying, we obtain:
\begin{equation}
\label{indifference_simplified}
\frac{(1-\lambda)R}{(1-F(c^*))H} = \lambda \frac{R}{F(c^*)H+M}-c^*. 
\end{equation}

 We reorder equation \eqref{indifference_simplified} as follows:
\begin{equation}
\label{indifference_simplified2}
    c^* =  \frac{\lambda R}{F(c^*)H+M} - \frac{(1-\lambda)R}{(1-F(c^*))H}.
\end{equation}
We note that the left hand side (LHS) of equation~\eqref{indifference_simplified2} is obviously increasing in $c^*$, while the right hand side (RHS) is decreasing in $c^*$. Indeed, the derivative of the RHS with respect to $c^*$ is $$-\frac{\lambda R F'(c^*)H}{(F(c^*)H+M)^2} - \frac{(1-\lambda)R  F'(c^*)}{(1-F(c^*))^2 H} < 0.$$ Since the LHS of~\eqref{indifference_simplified2} is increasing and is equal $0$ for $c^*=0$, and the RHS is decreasing in $c^*$, the necessary condition to have a solution to the equation is that the RHS is positive for $c^*=0$. That is, we have the condition $$\frac{\lambda R}{F(c^*)H+M} - \frac{(1-\lambda)R}{(1-F(c^*))H} >0,$$ which is, for $c^*=0$, equivalent to the condition in the theorem,
\begin{equation*}
    \lambda >\frac{M}{H+M}.
\end{equation*}

The indifference condition of the equilibrium of equation~\eqref{indifference_simplified2} to have an internal solution is obtained by taking $c^*=T$. In this case, the LHS has to be larger than the RHS, which always holds for $\lambda<1$, as the RHS is equal to $-\infty$.




To establish uniqueness, suppose that $\lambda >\frac{M}{H+M}$. As shown above, if we focus on threshold equilibria, there exists a unique equilibrium characterized with the cost level $c^*$.
Suppose that an equilibrium exists which is not of the threshold type. Without loss of generality, assume two cost levels $c_1$ and $c_2$ with $c_2 >c_1$, with the following property. A agent with cost $c_2$ will run a pool, while a agent with $c_1$ will delegate.
Hence, for the first agent, it must hold that
\begin{equation*}
     \frac{\lambda R}{F(c^*)H+M}-c_2 > \frac{(1-\lambda)R}{(1-F(c^*))H},
\end{equation*}
while for the second agent, we must have the opposite inequality, that is,
\begin{equation*}
     \frac{\lambda R}{F(c^*)H+M}-c_1 < \frac{(1-\lambda)R}{(1-F(c^*))H}.
\end{equation*}
Together, this implies that 
\begin{equation*}
    c_2 < \frac{\lambda R}{F(c^*)H+M} - \frac{(1-\lambda)R}{(1-F(c^*))H} < c_1,
\end{equation*}
which contradicts the assumption $c_2 >c_1$. Hence, any equilibrium will be of threshold type.
\end{proof}

Since all equilibria are of the threshold-type, we simply refer to threshold-type equilibria as ``equilibria". 

 The interpretation of the lower bound condition on $\lambda$ in the theorem is straightforward. To have a positive measure of pools owned by honest agents, the share of the reward for the pool owners should be higher than the share of malicious agents in the whole system. As long as $\lambda$ satisfies this condition, we have a unique equilibrium of the game $\mathcal{G}$. From the proof of Theorem~\ref{lower_bound_lambda}, it is straightforward to see that if $\lambda = \frac{M}{H+M}$, then the only solution to the indifference condition is $c^*=0$ and hence, all honest agents will delegate and malicious agents control all stakes. If $\lambda < \frac{M}{H+M}$, then there exists no equilibrium solution, where a positive measure of honest agents run pools.



\section{Optimal Reward Design}\label{sec: reward_design}

In this section, we use the framework developed in the previous sections to design blockchains that maximize security or, alternatively, maximize welfare. 

\subsection{Maximal Blockchain Security}

We obtain the equilibrium solution of \eqref{indifference_simplified} for a given $\lambda$ by simply solving for $c^*$. We denote it by $c^*(\lambda)$. The inverse function is denoted by $\lambda(c^*)$, and can be trivially found from~\eqref{indifference_simplified}. Namely,

\begin{equation}\label{inverselambda}
    \lambda(c^*) = \frac{c^* + \frac{R}{(1-F(c^*))H}}{\frac{R}{F(c^*)H+M}+\frac{R}{(1-F(c^*))H}}.
\end{equation}

$\lambda$ is a designer's variable. 

We assume in this section that the probability that the blockchain operates correctly is given by: 

\begin{equation*}
P_c(c^*):=\frac{P(c^*)}{P(c^*)+M}.
\end{equation*}

That is, the probability that the next block consists of correct transactions is equal to a share of honest pool runners. This is a simple formulation, reflecting that the next block writer is chosen uniformly at random. However, our analysis holds qualitatively for any probability function that is increasing in the share of honest agents and may reach 1 if a sufficient share of honest agents is achieved.

The first goal of the designer is to maximize the share of honest agents running pools, that is, to maximize $P(c^*)$. The probability that the blockchain is run correctly depends on $P$, which is increasing in $c^*$. 
Increasing $\lambda$ has two effects on the honest agents' decision. 
First, it motivates an agent to run a pool, as a greater  share of the rewards is allocated to the owner of the pool. On the other hand, since a higher $\lambda$ motivates many agents to run pools, there are many pools, and therefore, lower chances for each of them to win the reward. However, we obtain the following result:

\begin{proposition}
\label{prob_maximization}
The fraction of honest agents running a pool is maximized for $\lambda=1$.
\end{proposition}

\begin{proof}
Note that the RHS of~\eqref{indifference_simplified2} is increasing in $\lambda$. By increasing $\lambda$, we have to increase $c^*$ to have equality, as the RHS is decreasing in $c^*$. That is, if $\lambda_1\leq \lambda_2$, then $c^*(\lambda_1)\leq c^*(\lambda_2)$. 
Taking the maximum value $\lambda=1$ transforms equation~\eqref{indifference_simplified2} into 
\begin{equation}\label{lambdaone}
    c^* = \frac{R}{F(c^*)H+M}. 
\end{equation}
The solution to this equation maximizes the share of honest validators. 
\end{proof}

We note that by setting $\lambda=1$, in the threshold equilibrium of game $\mathcal{G}$, we do replicate the levels of honest and malicious stakes involved in transaction validation which would arise in the simple game without delegation and no staking pools. In this game, instead, honest agents are allowed to either validate transactions or abstain. In such a game, the indifference condition of the threshold equilibrium corresponds to $c^*=r$, equivalent to $\lambda=1$ in the pool formation game. Yet, with $\lambda=1$ and pool formation, all returns from validation are channelled to staking pool owners while delegators receive nothing. This is a concern for the future evolution of the blockchain since stake holding may be more and more concentrated on pool owners.   

We will see next that the solution $\lambda=1$ does not maximize social welfare, and does not distribute the rewards on honest agents that have high costs of running a pool either. 


\subsection{Welfare Optimal Reward Schemes}

In this section, we consider the alternative objective the Blockchain Designer may pursue, namely taking into account that achieving maximal security may involve large costs, as too many honest agents with high costs participate in the validation process.
For the alternative objective, we normalize the returns from a successful operation of the blockchain per honest agent to one and express the costs relative to these returns. To quantify the gains and losses, we introduce social welfare of the game $\mathcal{G}$.

\begin{definition}[Social Welfare]
Social welfare of the game $\mathcal{G}$ is defined as
\begin{equation}\label{welfare_definition}
    W=\frac{P(c^*)}{P(c^*)+M}H - P(c^*)\mathbb{E}[X|X<c^*].
\end{equation}
\end{definition}
The social welfare is calculated as the probability that the blockchain runs correctly, times the measure of honest agents,  minus average pool running costs incurred by honest agents.

Increasing $\lambda$ has two competing effects on  social welfare. First, it increases the number of honest agents who create their own pools. Therefore, the likelihood that malicious agents will write the next block is decreasing. Second, increasing the number of honest agents who create their own pools wastes a lot of costs of running pools. These two effects work against each other. In the following, we show that for a wide class of distribution functions, the social welfare optimum value is not polar. 
We obtain the following result:   

\begin{theorem}
Let the cost distribution function satisfy $F''(c)\leq 0$ for any $c\in [0,T)$. If $HM\geq(H+M)^2T$, then the optimum value maximizing welfare is $c^*=T$. On the other hand, if $HM<(H+M)^2T$, then there is a unique optimum value of $c^*$ that maximizes the social welfare.
\end{theorem}

\begin{proof}
The welfare~\eqref{welfare_definition} can be rewritten as 

\begin{align*}
    W=\frac{F(c^*)H^2}{F(c^*)H+M} - F(c^*)H\mathbb{E}[X|X<c^*] = \frac{F(c^*)H^2}{F(c^*)H+M} - F(c^*)H\frac{\int_{0}^{c^*}xf(x)d(x)}{F(c^*)}.
\end{align*}
The derivative of $W$ with respect to $c$ is equal to: 

\begin{align*}
    W' = \frac{f(c)H^2M}{(F(c)H+M)^2}  - H\cdot cf(c). 
\end{align*}
The second derivative is equal to:

\begin{align*}
    W'' = \frac{f'(c)H^2M(F(c)H+M) - 2 f(c)^2 H^3M}{(F(c)H+M)^3} - Hf(c) - cHf'(c). 
\end{align*}
We note that given $F''(c)<0$, the second derivative is always negative, that is, the welfare function is concave. It is easy to verify that $W'(0)\geq0$. That is, if $W'(T)<0$, which is equivalent to $HM<(H+M)^2T$, we have a unique optimum solution. If, on the other hand, $W'(T)\geq0$, equivalent to $HM\geq(H+M)^2T$, then the optimum is achieved in the point $c^*=T$. 
\end{proof}

Note that the uniform distribution function satisfies the condition $F''(c)\leq 0$. In fact, all distribution functions of the type $F(c)=c^{\alpha}$, where $\alpha \leq 1$ satisfy the condition of the theorem.

\section{Uniform Cost Distribution}\label{sec: numerical}

In this section, we analyze the equilibria that maximize blockchain security and calculate the welfare optimal values for $c^*$ and $\lambda$ for the case when costs are distributed uniformly on $[0,1]$. Furthermore, we study for the uniform distribution how rewards for malicious agents can be minimized.

\subsection{Security Maximization}\label{sec: unifCost: security max}
To maximize the share of honest staking pools, we have to solve equation \eqref{lambdaone}. Hence, for uniform distribution on the interval $[0,1]$, that is, $F(c^*)=c^*$, we have to solve the following quadratic equation:
\begin{equation}\label{quadratic_eq}
    c^2H + cM - R = 0.
\end{equation}
The positive solution is given by 
\begin{equation}\label{solforlambda1}
    c^*=\frac{-M + \sqrt{M^2+4H R}}{2H}.
\end{equation}
If $H=1$ and $M=0.4$ and $R=1$, then the solution to~\eqref{quadratic_eq} is $c^*=\frac{-1+\sqrt{26}}{5}\approx 0.82$. This solution maximizes the share of honest staking pools. That is, approximately $82\%$ of honest agents will run pools in the equilibrium. 
More numerical values are provided in Section \ref{sec: unifCost: numerical values}.

\subsection{Welfare Maximization} \label{sec: unifCost: welfare max}

For uniform distribution $F$, the welfare measure \eqref{welfare_definition} simplifies to 
\begin{equation*}
    W = H \left( \frac{c^*H}{c^*H+M}- \frac{1}{2}c^{*2} \right),
\end{equation*}
where we used that $P(c^*) = F(c^*)H=c^*H$ and 
$$E[X|X<c^*]=\int_{0}^{\infty}xf(x|x<c^*)dx  = \frac{\int_{0}^{c^*}xf(x)dx}{F(c^*)} = \frac{\frac{1}{2} c^{*2}}{c^*} =\frac{c^*}{2}. $$


The derivative of $W$ with respect to $c^*$ is
\begin{equation*}
    W'(c^*)=\frac{H^2M}{(c^*H+M)^2}-c^{*}H.
\end{equation*}
Note that $W'(c^*)=0$ has three solutions. Furthermore, $W''(c^*) = -\frac{2H^3M}{(c^*H+M)^3}-H < 0 $ for any non-negative $c^*$. This means that the real-valued extremum (which is in the interval $[0,1]$) of $W$ is a maximum. By simply solving $W'(c^*)=0$, we obtain the following real-valued solution which maximizes $W$: 
\begin{align} \label{csol}
    c^* =  \frac{(1+i\sqrt{3}) M^2}{2^{2/3} 3 Z} + \frac{(1-i\sqrt{3}) Z}{6 \sqrt[3]{2} H^2}- \frac{2M}{3H},
\end{align}
where $Z = \sqrt[3]{-27H^5M-2H^3M^3 + \sqrt[3]{3} \sqrt{27H^{10}M^2+4H^8M^4}}$.
The welfare function is maximized at $c^*$. 
Numerical examples follow in the subsequent section.








\subsection{Numerical Illustrations}\label{sec: unifCost: numerical values}
In this subsection, we provide detailed security maximizing (see Section \ref{sec: unifCost: security max}) and welfare maximizing (see Section \ref{sec: unifCost: welfare max}) equilibrium values for three sets of parameters in Table~\ref{tab:my_label} and Table~\ref{tab:my_label2}.
In Table~\ref{tab:my_label} the values for $c^*$ are given by Equation \eqref{solforlambda1}. In Table~\ref{tab:my_label2} the values for $\lambda$ and $c^*$ are given by Equations \eqref{inverselambda} and \eqref{csol}, respectively.

\begin{table}[htbp]
    \centering
    \begin{tabular}{|c|c|c|c|c|c|c|}
        \hline
         $H$ & $M $& $R $& $\lambda$ & $c^*$&$P (=c^*H)$&$W$ \\
         \hline \hline
         1& 0.5 & 1 &  1 &  0.78&0.78& 0.305\\
         1& 0.4 & 1 &  1 &  0.82&0.82& 0.336\\
         1& $\frac{1}{3}$ & 1 &  1 &  0.85 & 0.85 & 0.359\\
         \hline
    \end{tabular}
    \caption{Maximizing the measure of honest staking pools: $P$.}
    \label{tab:my_label}
\end{table}

\begin{table}[htbp]
    \centering
    \begin{tabular}{|c|c|c|c|c|c|c|}
        \hline
         $H$ & $M$ & $R$ & $\lambda $ & $c^* $&$P (=c^*H)$&$W$ \\
         \hline \hline
         1& 0.5 & 1 &  0.83 &  0.5 & 0.5 & 0.375\\
         1& 0.4 & 1 &  0.8 &  0.497 & 0.497 & 0.431\\
          1& $\frac{1}{3}$ & 1 &  0.77 &  0.491 & 0.491 & 0.475\\
         \hline
    \end{tabular}
    \caption{Maximizing social welfare: $W$.}
    \label{tab:my_label2}
\end{table}




For example, consider the case where $M=0.4$. Social welfare is maximized for cost $c^*=0.497$. That is, it is quite different from $c^*=0.82$, which is the value for maximizing the honest agents running a pool. 
On the other hand, $\lambda=0.8$, the value optimizing welfare is also quite different from $\lambda=1$, the value maximizing the share of honest agents running their own pools.

\subsection{Malicious Reward Minimization}\label{sec:malicious_reward}

The share of rewards received by malicious agents as a function of $\lambda$ is calculated in the following way:

\begin{equation}\label{malicious_rewards}
    \mu(\lambda):=\frac{M\lambda}{F(c^*)H+M}. 
\end{equation}

To show the main result of this section, we first show a lemma that holds for any cost distribution function: 

\begin{lemma}\label{lemma:generic}
Rewards for malicious agents when $\lambda=\frac{M}{H+M}$ are lower than rewards for $\lambda=1$, for any cost distribution function $F$.
\end{lemma}

\begin{proof}
In the first case, where $\lambda=\frac{M}{H+M}$,  we have $c^*=0$, and therefore, using~\eqref{malicious_rewards}, rewards are equal to $$\mu\left(\frac{M}{H+M}\right) = \frac{M}{0+M}\lambda = \frac{M}{H+M}.$$ In the second case, rewards are equal to $$\mu(1) = \frac{M}{F(c^*(1))H+M}\lambda = \frac{M}{F(c^*(1))H+M}.$$ Since $F(c^*(1))\leq 1,$ the lemma is proved. 
\end{proof}

Next, we show the main result of this section: 

\begin{proposition}\label{reward_minimization}
Rewards to malicious agents $\mu(\lambda)$ are minimized for $\lambda=\frac{M}{H+M}$ for a uniform cost distribution function on the interval $[0,1]$. Moreover, $\mu(\lambda)$ is increasing or first increasing and then decreasing on the interval $[\frac{M}{M+H}, 1]$. 
\end{proposition}

\begin{proof}
We verify the derivative of the share of rewards received by malicious agents with respect to $c^*$ for a uniform cost distribution. It is given by,

\begin{equation}\label{derivative_rewards}
    \frac{d \frac{M\lambda}{F(c^*)H+M}}{dc^*} = \frac{M \lambda'(c^*) (c^*H+M)-HM\lambda}{(c^*H+M)^2}.
\end{equation}
Note that the denominator is always positive. In Equation \eqref{inverselambda}, $\lambda$ is given as a function of $c^*$. The derivative is
\begin{equation*}
    \lambda'(c^*)= \frac{1}{R(H+M)} (-3{c^*}^2H^2 + 2 c^* (H^2-HM) + HR + MH).
\end{equation*}
After plugging this into Equation \eqref{derivative_rewards}, the numerator of \eqref{derivative_rewards},  $M \lambda'(c^*) (c^*H+M)-HM\lambda$, is given by 
\begin{equation*}
    \frac{-2H^3 M {c^*}^3  + H^2{c^*}^2 (HM-4M^2) + 2H c^* (H M^2 - M^3) + H M^3}{R(H+M)}.
\end{equation*}
Again, note that the denominator is always positive and hence, we only consider the numerator, $-2H^3 M {c^*}^3  + H^2{c^*}^2 (HM-4M^2) + 2H c^* (H M^2 - M^3) + H M^3$. For $c^*=0$ (equivalent to $\lambda=M/(H+M)$), it is positive, and for $c^*=\frac{-M+\sqrt{M^2+4HR}}{2H}$ (equivalent to $\lambda=1$), it is, depending on $H,M$ and $R$, either positive or negative.

As increasing (decreasing) $\lambda$ yields increasing (decreasing) $c^*$ and vice versa, \eqref{derivative_rewards} translates easily. 
Therefore, given $\mu(\frac{M}{H+M})<\mu(1)$ from Lemma~\ref{lemma:generic}, the minimum value is achieved in the point $\lambda=\frac{M}{H+M}$. 
From the observation $\mu(\frac{M}{H+M})<\mu(1)$, we see that the derivative cannot  always be negative. Therefore, it is either always positive, implying that $\mu(\lambda)$ is increasing on the whole interval $[\frac{M}{H+M},1]$, or the derivative is first positive and then negative, implying that $\mu(\lambda)$ is first increasing and then decreasing. 
 \end{proof}

This result suggests that rewards to malicious agents are minimized in the corner case, where only malicious parties run pools. In that case, however, the probability that the blockchain functions correctly is equal to $0$. Therefore, from the blockchain security perspective, designer needs a higher $\lambda$. Our result suggests that once $\lambda$ is large enough ($\lambda\geq t_{\lambda}$) for the blockchain security to cross the required threshold specified by a system, the designer only needs to verify endpoints of the interval $[t_{\lambda},1]$, for minimizing rewards to malicious agents.   

We next provide numerical examples for the share of rewards that malicious agents receive for a uniform cost distribution when $\lambda=1$ and $\lambda<1$.

First, for $\lambda = 1$, the malicious agents' reward share is $\frac{M}{P+M}$, with $P = F(c^*)H$, where $c^*$ is given by \eqref{solforlambda1}.
For $\lambda=1, H=1, R=1$ the share of rewards that malicious agents receive is given by
\begin{equation*}
    \frac{2M}{M+\sqrt{M^2+4}}.
\end{equation*}
We summarize some values in Table~\ref{tab: share for malicious for lambda =1}.
\begin{table}[htbp]
    \centering
    \begin{tabular}{|c|c|c|c|}
        \hline
         $H$ & $M $& $R $&  Share of Reward \\
         \hline \hline
         1& 0.5 & 1 &   0.390388\\
         1& 0.4 & 1 &   0.327922\\
         1& $\frac{1}{3}$ & 1 & 0.282376\\
         \hline
    \end{tabular}
    \caption{Share of rewards for malicious agents when $\lambda=1$.}
    \label{tab: share for malicious for lambda =1}
\end{table}

Second, for $\lambda < 1$, the malicious reward share is $\frac{M}{P_\lambda+M}\lambda$, with $P_\lambda = F(c^*)H$, where $c^*$ is the real-valued non-negative solution to \eqref{indifference_simplified}. We first solve the indifference equation as follows:
With uniform distribution, we obtain from \eqref{indifference_simplified},
\begin{equation}
\label{indifference_eq}
    0=H^2 {c^*}^3 + (MH-H^2){c^*}^2 - (RH + MH)c^* + \lambda R(H+M) - RM.
\end{equation}
To find the roots of this cubic equation, we can write \eqref{indifference_eq} as $0=a{c^*}^3+b{c^*}^2+c{c^*}+d$, where
\begin{align*}
    a&=-H^2,\\
    b&=H(H-M),\\
    d&=H(M+R),\\
    e&=-\lambda R (H+M) + MR.
\end{align*}
Let
\begin{align*}
    \Delta_0 &= b^2-3ad \\&= H^2((H-M)^2 + 3H(M+R)),\\
    \Delta_1 &= 2b^3-9abd+27a^2e \\&= H^3(2(H-M)^3 + 9H(H-M)(M+R) + 27HMR) - \lambda \cdot 27H^4R(H+M).
\end{align*}
Then, we define
\begin{equation}\label{cfromlambda}
    C = \sqrt[3]{\frac{\Delta_1 + \sqrt{\Delta_1^2 - 4 \Delta_0^3}}{2}},
\end{equation}
and the three roots of \eqref{indifference_eq} are given by 
\begin{equation}\label{sollambda<1}
    c^*_k = -\frac{1}{3a}\left( b + z^k C + \frac{\Delta_0}{z^k C} \right),
\end{equation}
where $z = \frac{-1 + i\sqrt{3}}{2}$ and $k=0,1,2$. 

We summarize some values in Table~\ref{tab: share for malicious for lambda <1}.
\begin{table}[htbp]
    \centering
    \begin{tabular}{|c|c|c|c|c|}
        \hline
         $H$ & $M $& $R $& $\lambda$ & Share of Reward \\
         \hline \hline
         1& 0.5 & 1 & 0.99 &  0.395647\\
         1& 0.5 & 1 & 0.9 &  0.414172\\
         1& 0.5 & 1 & 0.8 &  0.416164\\
         1& 0.5 & 1 & 0.7 &  0.409608\\
         1& 0.5 & 1 & 0.6 &  0.396822\\
         1& 0.5 & 1 & 0.5 &  0.378318\\
         \hline
         1& 0.4 & 1 & 0.9 &  0.35305\\
         1& 0.4 & 1 & 0.8 &  0.357138\\
         1& 0.4 & 1 & 0.6 &  0.345172\\
         1& 0.4 & 1 & 0.5 &  0.331877\\
         1& 0.4 & 1 & 0.4 &  0.313697\\
         \hline
         1& $\frac{1}{3}$ & 1 & 0.9 &  0.307422\\
         1& $\frac{1}{3}$ & 1 & 0.5 &  0.294599\\
         1& $\frac{1}{3}$ & 1 & 0.4 &  0.28047\\
         1& $\frac{1}{3}$ & 1 & 0.3 &  0.261626\\
         \hline
    \end{tabular}
    \caption{Share of reward that malicious agents receive.}
    \label{tab: share for malicious for lambda <1}
\end{table}

\section{Return Competition}\label{sec: competition}

In this section, we reconsider the staking pool formation game. Instead of the Blockchain Designer, we allow that pool runners choose both their own levels of rewards and the rewards they want to distribute to delegators. A pool owner $i$ sets his/her own $\lambda_i$. This game is a variant of the game $\mathcal{G}$ studied so far in the paper. 

In the first part of this section, we allow free return competition, that is, pool owner $i$ can choose any $\lambda_i \in [0,1]$. We denote this game by $\mathcal{G}_0$. Hence, the game unfolds as follows: 

\begin{itemize}
    
    \item[ ] Stage 1: Agents (both honest and malicious) either decide to form a staking pool or not. Agents who decide to become a pool owner obtain an identification number, denoted by $i$, and set $\lambda_i$.
    
    \item[ ] Stage 2: Agents who did not register for a staking pool decide whether to delegate their stake to some staking pool or to remain idle. 
    
    \item[ ] Stage 3: The blockchain runs, validation takes place (or not), and rewards are distributed. 
\end{itemize}

We obtain the following result:

\begin{proposition}\label{bad_equilibrium}
In any equilibrium of the game $\mathcal{G}_0$, malicious agents control all stakes involved in transaction validation and the blockchain is disrupted. 
\end{proposition}

\begin{proof}
Suppose that there exists an equilibrium in which honest agents run staking pools. Since running such staking pools is costly and there is zero measure of honest agents with zero costs, such an equilibrium necessarily must involve that the minimal value of all offered values  $\lambda_i$ by staking pool owners, denoted by  $\hat{\lambda}$, must be positive. Otherwise, honest agents are better off by delegating their stakes. Note that the last statement holds, since an individual honest agent has no influence on the probability that the blockchain operates correctly by his/her decision whether to run a staking pool or to delegate,

However, every malicious agent has an incentive to deviate and to set a lower value of $\hat{\lambda}$ for his/her own staking pool in order to attract more delegators, thereby making staking pools for honest agents unattractive. Hence, all honest agents delegate. This is a contradiction that honest agents run staking pools.
Hence, in any equilibrium,  malicious agents control all stakes involved in transaction validation and the blockchain is disrupted..  
\end{proof}

In the second part of this section, agents are only allowed to choose their corresponding $\lambda_i$ from the interval $[\bar{\lambda}, 1]$, where $\bar{\lambda}>\frac{M}{H+M}$. We denote this game by $\mathcal{G}_{\bar{\lambda}}$. Invoking  standard Bertrand competition logic, we obtain the following result in this case: 

\begin{proposition} 
The equilibrium of the extended game $\mathcal{G}_{\bar{\lambda}}$ is the same as the equilibrium of game $\mathcal{G}$.
\end{proposition}

\begin{proof}
First, we note that in the equilibrium, all malicious agents choose the lowest possible level $\bar{\lambda}$. If honest agents choose any $\lambda$ that is strictly larger than $\bar{\lambda}$, then nothing is delegated to them. Therefore, they also choose the same $\bar{\lambda}$. 
\end{proof}

That is, imposing a lower bound on $\lambda$ also guarantees the same upper bound. This adds to the robustness of the result obtained in Theorem~\ref{lower_bound_lambda} and offers a way to implement the equilibrium solution. On a practical side, the blockchain system does not need to force the agents to have the same level of rewards. Rather, they reach it through rational play. Note that setting any lower bound $\bar{\lambda}\leq \frac{M}{H+M}$ would result in the same (bad) equilibrium obtained in Proposition~\ref{bad_equilibrium}. 

\section{Extensions}\label{sec: extensions}
In this section, we study two extensions of our basic model. In first extension, we introduce a cost of delegation and in the second extension, we study the case when rewards are realized endogenously. Finally, we outline how the analysis has to be performed in a discrete framework and how we can recover results as the one in Theorem~\ref{lower_bound_lambda} in a discrete framework. 

\subsection{Costly Delegation}\label{sec: costlydelegation}

In this subsection, we include costs of delegation. Arguably, such costs are non-zero, since agents have to obtain the knowledge how to delegate their stakes safely to staking pools, which lock-up conditions are attached to such operations and how to observe the returns from delegation in comparison to other alternatives. This takes time and involves opportunity costs. We denote the costs of delegation by $c_d >0$. 

Costs for delegation introduce new trade-offs, as large values of $\lambda$ may motivate honest agents to run pools but simultaneously discourage honest agents from delegating to pools, as they may simply stay idle. We will see that staking pools and suitably chosen sharing parameters for $\lambda$ may increase blockchain security in such cases.

To prepare the analysis, let us first assume that the strategy set only consists of running a pool and delegating. We denote this game by $\mathcal{G}_d$. The indifference condition in this game is the same as~\eqref{indifference_simplified2}, with only one difference---the RHS has the additional summand $c_d$. 

\begin{equation}
    \label{indifference_costly_delegation}
        c^* =  \frac{\lambda R}{F(c^*)H+M} - \frac{(1-\lambda)R}{(1-F(c^*))H}+c_d.
\end{equation}

We obtain the following auxiliary result: 
\begin{proposition}
For any value of $\lambda > \frac{MR - c_dMH}{HR+MR}$, the indifference cost level of the equilibrium solution of the extended game $\mathcal{G}_d$ is higher than the cost level of the equilibrium solution of $\mathcal{G}$. 
\end{proposition}

\begin{proof}
The proof is analogous to the proof of Theorem~\ref{lower_bound_lambda}. The additional term $c_d$ increases the RHS of~\eqref{indifference_costly_delegation}, and therefore, the intersection of the LHS and RHS curves corresponds to higher $c^*$. 
By plugging in $c^*=0$, and requiring LHS is lower than RHS, and simplifying we obtain: 
\begin{equation*}
    \lambda > \frac{MR - c_dMH}{HR+MR}.  
\end{equation*}
\end{proof}

Note that the lower bound on $\lambda$ for a positive equilibrium in the game $\mathcal{G}_d$ is lower than the lower bound obtained in Theorem~\ref{lower_bound_lambda}. However, in this game, it might be that the expected utility of running a pool minus the expected utility of delegation is negative, but adding $c_d$ makes the RHS of~\eqref{indifference_costly_delegation} positive. In that case, individual rationality of pool running agents is violated. Therefore, if we consider an extended game $\mathcal{G}_d^i$, in which honest agents are allowed to be idle, then those with high costs would choose to stay idle instead of delegating. This also changes the behavior of agents that run  pools. In particular, they stay idle. That is, the strategy sets of agents are extended, and also the definition of the threshold equilibrium takes a new strategy into account. Namely, in the indifference condition, only positive values are compared, as the strategy to stay idle yields zero utility. Formally, we obtain the following result: 

\begin{proposition}
For $\lambda\in [\frac{MR - c_dMH}{HR+MR}, \frac{M}{M+H}]$, the equilibrium solution of the game $\mathcal{G}_d^i$ is lower than the one of $\mathcal{G}_d$. Below the threshold, all honest agents run pools and above it all honest agents stay idle. 
\end{proposition}

\begin{proof}
When $\lambda\in [\frac{MR - c_dMH}{HR+MR}, \frac{M}{M+H}]$, the expected gain from running a pool minus the expected gain from delegation is negative: $\frac{\lambda R}{F(c^*)H+M} - \frac{(1-\lambda)R}{(1-F(c^*))H}<0$. 
Agents with cost $c \geq c^*$ stay idle. Agents with cost $c < c^*$ are divided in two groups. Namely, there exists a $c' \in (0,c^*)$ such that all agents with cost $c < c'$ run a pool and others stay idle. The threshold $c'$ is a solution of the following equation: $\frac{\lambda R}{F(c')H+M}-c' = 0$ and does always exist.
\end{proof}

Next, we ask  if costly delegation can help to increase the level of honest agents running a pool. Note that $\lambda = 1$ maximizes the level without costly delegation. From equation~\eqref{lambdaone}, we find the level of $c^*$,  such that agents with cost lower than this threshold will run a staking pool and those agents with costs higher than the threshold will stay idle, since delegation is costly. However, if we decrease $\lambda$ in such a way that the expected return from delegation compensates the delegation cost $c_d$, then all honest agents who were staying idle will delegate or run a staking pool. 

We obtain the following result: 

\begin{proposition}
With costly delegation, the highest possible blockchain security level is always lower than the one in the benchmark game $\mathcal{G}$.\footnote{In knife-edge cases, both security levels can be the same. }
\end{proposition}

\begin{proof}
We first look at the decision of agents whether to delegate or to stay idle. For values of $\lambda$ that satisfy the following inequality:

\begin{equation}\label{second_condition_cost}
    \frac{(1-\lambda(c^*))R}{(1-F(c^*))H} \geq c_d,
\end{equation}

honest agents delegate instead of staying idle. $\lambda_d(c^*)$ is calculated from~\eqref{indifference_costly_delegation}, that is, from the indifference condition of the game $\mathcal{G}_d$. That is, $c^*$ should satisfy both conditions~\eqref{indifference_costly_delegation} and~\eqref{second_condition_cost}.

Recall that in the benchmark case $\lambda=1$ and $c_d=0$, the equilibrium cost level $c^*$ is a solution to the following equation: 

\begin{equation*}
    c = \frac{R}{F(c)H+M}.
\end{equation*}

In the costly delegation game, however, $c^*$ solves the equation~\eqref{indifference_costly_delegation}, which together with~\eqref{second_condition_cost}, yields: 

\begin{equation*}
    c^* =  \frac{\lambda R}{F(c^*)H+M} - \frac{(1-\lambda)R}{(1-F(c^*))H}+c_d\leq \frac{\lambda R}{F(c^*)H+M}\leq \frac{R}{F(c^*)H+M}.
\end{equation*}

Given the monotonic decreasing property of $\frac{\lambda R}{F(c^*)H+M}$, we obtain that the equilibrium $c^*$ of the costly delegation game is always lower than the equilibrium level of game $\mathcal{G}$.
\end{proof}

\subsection{Endogenous Rewards}\label{sec: endogenous_rewards}
Throughout the paper, we assumed that rewards for writing the next block are exogenously given and are equal to a constant number $R$, no matter what fraction of honest agents runs pools and participates in the validation. 
In this section, we assume that rewards are realized only if the blockchain functions correctly, with probability one. This probability is calculated by the following formula:   

\begin{equation*}
P_c(c^*):=\min \biggl[1, \biggl(\frac{P}{P+M} + \theta \biggl)\biggl].
\end{equation*}

Here $\theta$ is a real number in $[0,\frac{1}{2}]$ that describes the tolerance of a system regarding the share of malicious agents it can handle without compromising network security. 

We say that full network security is achieved when \begin{equation*}
\frac{P}{P+M} \geq 1 - \theta.
\end{equation*}

Typically, for example in Byzantine-fault-tolerant protocols, $\theta$ is about $\frac{1}{3}$ in many consensus protocols (see~\cite{Lamport1982},~\cite{Abd2005},~\cite{David2017}~\cite{Aiyer2005}  and~\cite{Dinsdale2019}). 
Thus, in that case, if the fraction of staking pools run by honest agents is at least $\frac{2}{3}$, then full security is achieved. 

Requiring the probability of the blockchain security to surpass the threshold $1-\theta$ imposes a threshold on the cost of pool running in the equilibrium. We denote the corresponding game by $\mathcal{G}^e$ and the cost threshold by $c^{\theta}$. It is defined by the following equation: 

\begin{equation}\label{threshold_cost}
    c^{\theta} := \inf_c{\frac{F(c)H}{F(c)H+M}\geq 1-\theta.}
\end{equation}  

In this setting, to have a unique threshold equilibrium, we need 
\begin{equation}\label{crit_indiff_sol}
    c^{\theta} < \frac{\lambda R}{F(c^{\theta})H+M} - \frac{(1-\lambda)R}{(1-F(c^{\theta}))H}.
\end{equation}

We obtain a result similar to the one of Theorem~\ref{lower_bound_lambda}:

\begin{theorem}
There exists a unique threshold equilibrium $c^*>c^{\theta}$ to the game $\mathcal{G}^e$ if and only if \begin{equation}
\label{lambda_lowerbound_crit}
    \lambda >\frac{\frac{c^{\theta}}{R} (F(c^{\theta})H+M)(1-F(c^{\theta}))H + F(c^{\theta})H+M}{H+M}. 
\end{equation}
\end{theorem}

The proof is analogous to the proof of Theorem~\ref{lower_bound_lambda}. 

The results from the Section~\ref{sec: competition} on return competition are also translated directly in this setting. In particular, if the agents are allowed to set any $\lambda\in [0,1]$ as their own pool return, the blockchain security fails. We denote this game by $\mathcal{G}_0^e$. Formally, we obtain the following result: 

\begin{proposition}
In any equilibrium of the game $\mathcal{G}_0^e$ no honest agent runs a pool and the blockchain is disrupted. 
\end{proposition}

The proof is analogous of the proof of Proposition~\ref{bad_equilibrium} and exploits the fact that individual honest agents have zero measure. Therefore, any unilateral deviation by an agent does not affect the probability that the blockchain operates correctly.

\subsection{Discrete Case}\label{Discrete_Case}
 
In this subsection, we show how the analysis can be recast in the discrete framework---albeit in a much more complicated form. In particular, we show that the same lower bound condition on $\lambda$ for having positive threshold equilibrium as in Theorem~\ref{lower_bound_lambda} can be obtained if we assume that the number of honest agents is a large integer $n>0$, and replace the expected number of honest agents who run pools with $n F(c^*)$. The latter approximately holds by Chernoff concentration bounds for large enough $n$. This observation adds to the  robustness of the continuum approach, as it shows that the result obtained in Theorem~\ref{lower_bound_lambda} is not a byproduct of our assumption on having infinitely small agents. Rather,  this assumption allows us to derive clean results more easily.

Let $m\in \mathbb{N}$ denote the number of malicious agents.
By the assumption on $c^*$, we have $k+m$ pools, where $k$ is distributed as binomial random variable with parameters $n$ and $F(c^*)$. In expectation, every pool will obtain  $\frac{n-k}{k+m}$  delegated stakes. Hence, every pool has a total expected stake of $1+\frac{n-k}{k+m} = \frac{n+m}{k+m}$.
If a pool is chosen (i.e. it becomes the new proposer), then all its delegators receive a reward. In expectation, the probability that a pool is chosen is $\frac{1}{k+m}$ and the reward for all delegators who  delegated to this pool is $(1-\lambda)r$.

The pool owners always face  private cost $c^*$ and obtain the  reward $\lambda r$ with probability $\frac{1}{k+m}$.
The expected utility of the single delegator is $\frac{1}{k+m}(1-\lambda) \frac{r}{d}$, where $d$ is the number of delegators to that particular pool. 

Let $d_1,...,d_{k+m}$ be non-negative numbers, such that $\sum_i d_i = n-k$, where $d_i$ denotes the number of delegators to pool $i$.
The ex-post utility of a delegator of a baking pool $i$, when pool $i$ has $d_i$ delegators:
\begin{equation*}
    \frac{1+d_i}{n+m} (1-\lambda) \frac{r}{d_i},
\end{equation*}
and for the pool owner,
\begin{equation*}
    \frac{1+d_i}{n+m}\lambda r - c_i.
\end{equation*}
We model $d_i$ as a binomial random variable, $d_i \sim Bin(n-k, \frac{1}{k+m})$.

Next, we replace random variables with their means in the ex-post indifference equation, i.e., instead of $d_i$, we insert $E[d_i]=\frac{n-k}{k+m}$ and instead of $k$, we insert $E[k]=F(c^*)n$. Then, we obtain: 
\begin{align*}
    \frac{1+d_i}{n+m} (1-\lambda) \frac{r}{d_i} &= \frac{1+d_i}{n+m}\lambda r - c^*\\
      \frac{1+\frac{n-F(c^*)n}{F(c^*)n+m}}{n+m} (1-\lambda) \frac{r}{\frac{n-F(c^*)n}{F(c^*)n+m}} &= \frac{1+\frac{n-F(c^*)n}{F(c^*)n+m}}{n+m}\lambda r - c^* \\
     c^* &= \frac{1}{F(c^*)n+m}\lambda r - \frac{1}{n-F(c^*)n}(1-\lambda)r.
\end{align*}

The left hand side is equal to $0$ for $c^*=0$ and it is increasing in $c^*$. The right hand side is a decreasing function in $c^*$ and therefore, it should be non-negative for $c^*=0$, to have a positive solution in $c^*$. That is,
\begin{align*}
    \frac{1}{F(c^*)n+m}\lambda r \geq \frac{1}{n-F(c^*)n}(1-\lambda)r.
\end{align*}
For $c^* = 0$, we take into account that $F(0)=0$ for a distribution function $F$. Thus, we obtain:  
\begin{align*}
    \frac{1}{m}\lambda & \geq \frac{1}{n}(1-\lambda)\\
   \lambda &\geq \frac{m}{n+m}.
\end{align*}

\section{Conclusion}\label{sec: conclusion}
In this paper, we initiated the study of the formation of staking pools from game-theoretic and mechanism design perspectives. Our insights can help to design reward distribution rules that improve the blockchain security and fairness of reward distribution. 
This study might open many further research avenues on how staking pools can be designed optimally for blockchains. For instance, one might introduce quantity constraints such as a leverage constraint on staking pools which limits the share of delegators towards the pool owner. Whether such a constraint further improves the security of the blockchain is left for future research. One might also allow  the history of staking and running staking pools to play a role in dynamic versions of the game and thus, an agent's reputation to behave honestly may be taken into account in such staking pool formation games. Future research can also extend our model to arbitrary distribution of stakes across agents.

\bibliographystyle{apalike}
\bibliography{references.bib}

\end{document}